\documentclass[aps,twocolumn, longbibliography,superscriptaddress,floatfix,]{revtex4-2}

\usepackage{graphicx}
\usepackage{amsmath}
\usepackage{amssymb}
\usepackage{amsthm}
\usepackage{tikz}

\usepackage[colorlinks = true]{hyperref}
\usepackage{xcolor}
\definecolor{darkred}  {rgb}{0.5,0,0}
\definecolor{darkblue} {rgb}{0,0,0.5}
\definecolor{darkgreen}{rgb}{0,0.5,0}

\hypersetup{
  urlcolor   = blue,         % color of external links
  linkcolor  = darkblue,     % color of internal links
  citecolor  = darkgreen,    % color of links to bibliography
  filecolor  = darkred       % color of file links
}
\usepackage{cleveref}

\newcommand{\be}{\begin{equation}}
\newcommand{\ee}{\end{equation}}
\newcommand{\ba}{\begin{array}}
\newcommand{\ea}{\end{array}}
\newcommand{\bea}{\begin{eqnarray}}
\newcommand{\eea}{\end{eqnarray}}

\newcommand{\ra}{\rangle}
\newcommand{\la}{\langle}

\newcommand{\calH}{{\cal H }}

\newcommand{\calF}{{\cal F }}

\newcommand{\calQ}{{\cal Q }}

\newcommand{\calU}{{\cal U }}
\newcommand{\Cbb}{\mathbb{C}}

\newcommand{\Sbb}{\mathbb{S}}
\newcommand{\Ubb}{\mathbb{U}}
\newcommand{\Zbb}{\mathbb{Z}}

\newtheorem{dfn}{Definition}

\newtheorem{lemma}{Lemma}
\newtheorem{fact}{Fact}
\newtheorem{conj}{Conjecture}

\newtheorem{corol}{Corollary}

\begin{document}

%\title{Building Character:\\  Computing $S_n$-characters with Tensor-Networks and Quantum Circuits}

\title{Classical and quantum algorithms for  characters of the symmetric group}

\author{Sergey Bravyi}
\affiliation{IBM Quantum, IBM T.J. Watson Research Center, Yorktown Heights, NY 10598 (USA)}
\author{David Gosset}
\affiliation{Department of Combinatorics and Optimization and Institute for Quantum Computing, University of Waterloo}
\affiliation{Perimeter Institute for Theoretical Physics}
\author{Vojtech Havlicek}
\affiliation{IBM Quantum, IBM T.J. Watson Research Center, Yorktown Heights, NY 10598 (USA)}
\author{Louis Schatzki}
\affiliation{University of Illinois Urbana Champaign, Urbana, Illinois 61801 (USA)}

\begin{abstract}
%DG: "the group theory"-->"group theory"
Characters of irreducible representations are  ubiquitous  in group  theory. 
However, computing  characters of some groups such as the symmetric group $S_n$ is a challenging  problem  known to be \#P-hard in the worst case.
Here we describe a Matrix Product State  (MPS) algorithm for  characters of $S_n$. The algorithm computes an MPS  encoding all irreducible characters of a given permutation.
 It relies on a mapping from  characters of $S_n$  to quantum spin chains  proposed by
Crichigno and Prakash. We also provide a simpler derivation of this mapping.
%SBB: small changes in this sentence
We complement this result by presenting a $poly(n)$ size quantum circuit that prepares the corresponding MPS
obtaining an efficient quantum algorithm for certain sampling 
 problems based on characters of $S_n$. To assess classical hardness of these
problems we present a general reduction from strong simulation (computing a given probability)
to weak simulation (sampling with a small error).  
This reduction applies to any sampling problem with a certain granularity structure and may be of independent interest. 
\end{abstract}

\maketitle

\section{Introduction}

Representation theory provides a  powerful set of tools for analyzing  systems invariant under the action of some group
such as a molecular symmetry group  in chemistry or a gauge group in quantum field theory.
For example, the eigenvalue spectrum of a quantum system is usually computed  by decomposing the Hilbert space
 into a direct sum of irreducible representations of the symmetry group and diagonalizing each irreducible
block of the Hamiltonian separately. 

The character of an irreducible representation (irrep)  can be viewed as its fingerprint that uniquely identifies
the irrep and facilitates many computational tasks. 
In the case of abelian groups such as the symmetry group of a periodic lattice, characters
are  basis functions of the Fourier transform associated with a given frequency or momentum.
In general,  the character of a group element $g$
is defined as the trace of a matrix describing the action of $g$ for a given irrep.

Here we address the problem of computing characters of the symmetric group $S_n$, that is, the group of all possible permutations of $n$ objects.
Symmetric groups $S_n$ with small values of $n$ often emerge as molecular symmetries in chemistry. For example, $S_4$ and $S_5$ describe
molecules with tetrahedral and  icosahedral symmetries  such as the methane and the fullerene.
The group $S_6$ 
is isomorphic to the symplectic group 
describing the action of two-qubit Clifford gates commonly used in quantum computation~\cite{grier2020interactive}. 
Moreover,
Cayley's theorem asserts that any finite group is isomorphic to a subgroup of $S_n$~\cite{johnson1971minimal}.

%DG: "the runtime"-->"runtime"
Unfortunately, characters of $S_n$ lack a simple analytic formula. In fact, it was shown~\cite{hepler1996complexity,Ikenmeyer24} that the problem of computing 
the character of a given permutation $g\in S_n$ for a given irrep  of $S_n$ is \#P-hard, that is, at least
as hard as counting the number of solutions of an NP-hard problem.  While this leaves little hope for an efficient algorithm with runtime scaling polynomially with $n$,
one may ask how well can we
do in practice and whether it is possible to compute characters of $S_n$ for moderately large values of $n$, say 
$n\le 100$ ?
%SBB3: changed 50 to 100

To motivate this problem, we note that 
characters of $S_n$  with large values of $n$ describe  amplitudes of Laughlin states in the Fractional Quantum Hall Effect~\cite{dunne1993slater,di2017unified,di2015,di1994laughlin}.
Having access to characters of $S_n$ enabled computation of the electron density and correlation functions of Laughlin states~\cite{dunne1993slater}.
Another application is computing  integrals of high-degree polynomials over the unitary group~\cite{collins2006integration}.
Such integrals can be computed using Weingarten functions which depend on $S_n$ characters~\cite{collins2006integration}.
Algorithms for $S_n$ characters may also facilitate evaluation of kernel functions for solving ranking problems in machine learning~\cite{kondor2008group,kondor2010ranking}.

Our first contribution is an algorithm for computing characters of $S_n$ based on Matrix Product States (MPS).
These are tensor networks commonly used for simulating ground state properties of quantum spin chains~\cite{vidal2003efficient}.
A deep connection between  characters of $S_n$  and quantum spin chains was recently 
established by Crichigno and Prakash~\cite{crichigno2024quantumspinchainssymmetric}.
These authors showed that characters of a given permutation $g\in S_n$ associated with all possible irreps of $S_n$
can be read off as amplitudes of a certain quantum state $|\psi_g\ra$ describing $2n$ fermionic modes.
The state $|\psi_g\ra$ is defined in terms of current operators
\[
J_\ell = \sum_{i=1}^{2n-\ell} a_{i+\ell}^\dag a_i,
\]
where $\ell$ is an integer and $a_i^\dag$ ($a_i$) are fermionic creation (annihilation) operators. Specifically,
\[
|\psi_g\ra = \prod_{j=1}^{c(g)} J_{\nu_j} |1^n0^n\ra,
\]
where $c(g)$ is the number of cycles in $g$, $\nu_j$ is the length of the $j$-th cycle,
and $|1^n0^n\ra$ is the Fock basis state with the first $n$ modes occupied and remaining modes empty.
The order in the product does not matter since the current operators pairwise commute~\cite{crichigno2024quantumspinchainssymmetric}.
For each irrep $\lambda$ of $S_n$, the Crichigno and Prakash mapping gives a Fock basis vector $x_\lambda \in \{0,1\}^{2n}$ such that the
amplitude $\la x_\lambda |\psi_g\ra$ coincides with the character of $g$ associated with $\lambda$.
Section~\ref{sec:spin_chain} provides a simplified derivation of  Crichigno and Prakash mapping.
Our derivation may be slightly more accessible to physicists as it avoids the heavy machinery of non-commutative symmetric polynomials used in 
the original work~\cite{crichigno2024quantumspinchainssymmetric, fominGreeneNoncommutative98}.
%DG: Below and throughout the paper, I changed "Slater determinant" to "Slater determinant"
Instead, we apply an elegant  formula derived by 
Frobenius in 1900~\cite{Frob1900} to express the characters in terms of fermionic Slater determinants.
Rewriting this formula using the fermionic 
second quantization formalism reproduces the result of~\cite{crichigno2024quantumspinchainssymmetric}.

%DG:"with the bond dimension"-->"with bond dimension"
The standard Jordan Wigner transformation maps $2n$ fermionic modes to $2n$ qubits.
We show that  the  image of the current operator $J_\ell$  under this mapping is a Matrix Product
Operator (MPO) with bond dimension $\ell+2$.
Applying a sequence of MPOs to an initial product state is a common subroutine in quantum spin chain simulations
and a variety of highly optimized software libraries for this task is available  such as $\mathsf{quimb}$~\cite{gray2018quimb} or $\mathsf{mpnum}$~\cite{suess2017mpnum}.
Our problem is particularly well suited for MPO-MPS simulation because characters of $S_n$ are integer valued.
Accordingly, a small error introduced by truncating singular values in the intermediate MPS
can be tolerated. A detailed description of our algorithm and numerical experiments can be found in Section~\ref{sec:classical}.

As mentioned above, computing $S_n$ characters is  a \#P-hard problem.
Thus, neither classical nor quantum algorithms for this problem are expected to run in time $poly(n)$.
The main obstacle for the classical MPS algorithm is a rapid growth of entanglement  
which results in an exponentially large MPS bond dimension. A natural question is whether a quantum computer
can efficiently prepare the state $|\psi_g\ra$ ? Surprisingly, we show that the answer is YES.
To this end, we construct a quantum circuit of size $poly(n)$ that prepares the uniform superposition of all permutations  in the conjugacy class of $g$.
We show that applying the Quantum Fourier Transform (QFT) over $S_n$ to this superposition gives the desired state $|\psi_g\ra$, with a suitable normalization coefficient.
Importantly,  the QFT   over $S_n$ can be realized by a quantum circuit of size $poly(n)$, as was shown by Beals~\cite{beals1997quantum}.
%DG: Edited the next two paragraphs for clarity

Using a similar technique, we describe a $poly(n)$ sized quantum circuit that, for a given irrep $\lambda$ of $S_n$, prepares a quantum state $|\omega_{\lambda}\rangle$ which is a superposition of all permutations $g\in S_n$ with amplitudes proportional  to the character of $g$ for irrep $\lambda$. By preparing and then measuring either $|\psi_{g}\rangle$ or $|\omega_{\lambda}\rangle$, a quantum computer can efficiently {\em sample} either irreps or permutations with probability proportional to the character squared. 

%We use similar technique to construct a $poly(n)$ quantum circuit that prepares a superposition of all permutations $g\in S_n$
%where amplitudes are proportional to the character of $g$ for a given irrep. 
%As a consequence, a quantum computer can efficiently {\em sample} 
%irreps or permutations with 
%the probability proportional to the  character squared. 

Finally, we ask whether or not these sampling problems admit efficient classical algorithms. 
%Our last contribution is an attempt at proving classical hardness of sampling problems based on characters of $S_n$.
To this end we give a reduction from weak simulation (sampling with a small statistical error)
 to strong  simulation (exact computation of probabilities)
for a broad class of sampling problems that posses  a certain granularity structure. 
Namely, if $P$ is a probability distribution on some set $\Omega$, we say that $P$ is {\em granular} if
$|\Omega|\cdot P(x)$ is integer for all $x\in \Omega$.  
Our weak-to-strong reduction applies to any granular (or approximately granular) distribution and does not require
the anti-concentration property which is commonly used in reductions from weak to strong simulation~\cite{aaronson2011computational,bremner2016average,hangleiter2023computational}.
Sampling problems based on characters of $S_n$ are shown to be 
 (approximately) granular due to the fact that $S_n$ characters are integer-valued. 
Using our reduction we show that classical hardness of sampling from the measured distribution of $|\psi_g\rangle$ for permutations $g$ with $O(\sqrt{n}/\log(n))$ cycles follows from a conjecture concerning average case hardness of computing the group characters associated with $g$. Establishing or refuting this conjecture is left for future work.

Similarly to characters of the symmetric group, we give an MPS algorithm for Kostka numbers. The first step of mapping Kostka numbers onto quantum spin chains was shown in ~\cite{crichigno2024quantumspinchainssymmetric}. We give a corresponding MPO construction and code for running this algorithm in the final secton of this paper.

%SBB: shortened this paragraph and added some more citations
{\em Related work.}
Most of the existing classical algorithms for computing $S_n$ characters use variants of a recursive formula known as
the Murnaghan–Nakayama rule~\cite{murnaghan1937characters,nakayama1940some}. Algorithmic versions of this formula can be found e.g. in~\cite{hepler1996complexity,eugeciouglu1985algorithms}.
%SBB: shortened the discussion of Jordan's work
An efficient classical algorithm 
that computes a small additive error approximation of 
$S_n$ characters normalized by the irrep dimension
is described  by Jordan~\cite{jordan2009fastquantumalgorithmsapproximating}. 
Classical and quantum 
algorithms for computing multiplicity coefficients
closely related to $S_n$ characters were studied in \cite{bravyi2024quantum, larocca2024quantumalgorithmsrepresentationtheoreticmultiplicities}. 
%SBB: this sounds too technical for a phyiscal journal
%Ref. \cite{larocca2024quantumalgorithmsrepresentationtheoreticmultiplicities} presents a similar algorithm to Jordan's for Kostka numbers, but it remains an interesting open questions whether there is a similar simulation strategy for Littlewood-Richardson, Kronecker and Plethysm coefficients.
% VH: OK :)

{\em Organization of the paper:}
Section~\ref{sec:background} provides a necessary background on group representation theory. 
Section~\ref{sec:spin_chain} gives a simplified derivation of  Crichigno and Prakash mapping
from characters of $S_n$ to quantum spin chains. Sections~\ref{sec:classical} and~\ref{sec:quantum}
contain our main results —
an MPS algorithm for characters of $S_n$ and a quantum circuit that prepares the corresponding MPS.
Sampling problems for granular probability distributions and  reductions from weak-to-strong simulation 
are discussed in 
Section~\ref{sec:granular}. An extension of our MPS algorithm to computing Kostka numbers is given in Section~\ref{sec:Kostka}.

\section{Group theory refresher}
\label{sec:background}

Suppose $G$ is a finite group. 
A $d$-dimensional representation of $G$ is a group homomorphism 
from $G$ to the unitary group $\Ubb(d)$. A representation  is  irreducible if the ambient Hilbert space $\Cbb^d$ contains no proper subspace
invariant under the action of all group elements. We use the shorthand irrep to denote an irreducible representation. 
It is standard to identify irreps that differ only by a  basis change in the ambient Hilbert space.
Then any finite group has a finite set of irreps.
Given a group element $g\in G$ and an  irrep $\lambda$, let 
$\rho_\lambda(g)$ be the unitary matrix corresponding to $g$.
The character of $g$ associated with an irrep $\lambda$ is defined 
as a complex number $\chi_\lambda(g) = \mathrm{Tr}(\rho_\lambda(g))$.
The fact that $\rho_\lambda$ is a group homomorphism implies that $\chi_\lambda(g)$ depends only 
on the conjugacy class of $g$, that is, $\chi_\lambda(g)=\chi_\lambda(hgh^{-1})$ for all $g,h\in G$.
Organizing all characters  into a matrix with rows labelled by irreps 
and columns labelled by conjugacy classes of $G$ gives the character table of $G$.

The symmetric group $S_n$ is the group
of permutations of $n$ objects, say integers $[n]=\{1,2,\ldots,n\}$.
Irreps and conjugacy classes of $S_n$ are labeled by partitions of $n$,
that is,  sorted sequences of positive integers 
$\lambda=(\lambda_1,\lambda_2,\ldots,\lambda_m)$
such that  $\lambda_1\ge \lambda_2\ge \ldots\ge \lambda_m$ and
$\sum_{j=1}^m \lambda_j= n$. 
Here $m$ is the number of parts in $\lambda$.
We shall write $\lambda \vdash n$ to indicate that $\lambda$ is a partition of $n$.
A permutation $g\in S_n$ is said to be  a  cycle of length $\ell$  if there exists a subset $J\subseteq [n]$ 
such that $|J|=\ell$ and $g$ cyclically shifts elements of $J$ while acting trivially on $[n]\setminus J$.
Any permutation $g\in S_n$ can be uniquely written as a composition of cycles with pairwise disjoint supports.
Denoting  the number of cycles by $c(g)$ and  the length of the $j$-th longest  cycle  by $\nu_j$ one obtains
a partition  of $n$ with $c(g)$ parts, namely, $\nu=(\nu_1,\nu_2,\ldots,\nu_{c(g)})$. 
The partition $\nu$ specifies the conjugacy class of $S_n$ that contains $g$.
Irreps of $S_n$ are also labeled by partitions of $n$. For example, a partition 
$\lambda=(n)$ with a single part
labels the trivial representation
while a partition
$\lambda=(1,1,\ldots,1)$ with $n$ parts
labels the sign representation.
The general recipe for constructing the unitary action $\rho_\lambda$ for a given partition $\lambda$ can be found in~\cite{etingof2011introduction}.
The explicit form of matrices $\rho_\lambda(g)$ is irrelevant for our purposes
since we are only interested in the characters $\chi_\lambda(g) = \mathrm{Tr}(\rho_\lambda(g))$.

\section{Building the character}
\label{sec:spin_chain}

We begin by reproducing the result of Crichigno and Prakash~\cite{crichigno2024quantumspinchainssymmetric}
that shows how to build a quantum spin chain encoding all characters $\chi_\lambda(g)$ of a given permutation $g\in S_n$.

Let $\lambda=(\lambda_1,\ldots,\lambda_m)$ and $\nu=(\nu_1,\ldots,\nu_{c(g)})$ be the partitions of $n$ describing the 
irrep $\lambda$ and the conjugacy class of $g$.
Consider a complex vector $z\in \Cbb^m$ and a polynomial
%DG: I think the dummy variable was used twice here and should only appear once. I changed one of the dummy variables to k
\be
P_\nu(z)=\prod_{j=1}^{c(g)} \left(\sum_{k=1}^m z_k^{\nu_j}\right).
\ee
Frobenius character formula~\cite{Frob1900} (see e.g. Theorem 4.47 in~\cite{etingof2011introduction}) states that $\chi_\lambda(g)$ coincides with the coefficient of a monomial 
\be
\label{z(lambda)}
z(\lambda)= \prod_{j=1}^m z_j^{\lambda_j + m - j}
\ee
in the polynomial
\be
\label{Q(z)}
Q(z) =P_\nu(z) \prod_{1\le i<j\le m} (z_i - z_j).
\ee
Let 
$\Sbb^1=\{x\in \Cbb\, :\, |x|=1\}$ be the unit circle in the complex plane
and  $\calF_1$ be a Hilbert space of functions $f\, : \, \Sbb^1\to \Cbb$
with the inner product 
\[
\la f|g\ra = \frac1{2\pi} \int_0^{2\pi} d\theta f(e^{i\theta})^* g(e^{i\theta})
\]
and an orthonormal basis of functions $\{f(x)=x^k\}_{k\in \Zbb}$.
Informally, $\calF_1$ describes a
particle moving on the circle  and $f(x)=x^k$ is a state with the angular  momentum $k$.
Consider now a many-body system composed of $m$ indistinguishable  particles
moving on the circle 
that obey fermionic statistics.
Fermionic first quantization principle stipulates that the
$m$-particle Hilbert space $\calF_m$ is the anti-symmetric subspace of $\calF_1^{\otimes m}$
and its  orthonormal basis  is formed by Slater determinants composed of $m$ distinct single-particle basis functions. Such Slater determinants
have a form $|\kappa\ra= D_\kappa(z)/\sqrt{m!}$, where
\[
D_\kappa(z) = \mathrm{det}\left[ \ba{cccc}
z_1^{\kappa_1} & z_1^{\kappa_2} & \ldots & z_1^{\kappa_m} \\
z_2^{\kappa_1} & z_2^{\kappa_2} & \ldots & z_2^{\kappa_m} \\
\vdots & \vdots & & \vdots \\
z_m^{\kappa_1} & z_m^{\kappa_2} & \ldots & z_m^{\kappa_m} \\
\ea
\right],
\]
and  $\kappa$ runs over strictly decreasing
integer sequences of length $m$ that represent momenta of the occupied single-particle states.
The Vandermonde determinant formula gives
\be
\label{Vandermonde}
\prod_{1\le i<j\le m} (z_i-z_j) = \sqrt{m!} |\tau\ra,
\ee
where
%DG: added period after equation
\[
\tau = (m-1,m-2,\ldots,1,0).
\]
%VH: Note that tau is the same as Frobenius "s hift"; Etingof calls it rho in 4.47
Let $P_\nu\,: \, \calF_m\to \calF_m$ be a linear operator that maps a function $\Psi(z)\in \calF_m$ to $P_\nu(z) \Psi(z)$.
Note that $P_\nu$ preserves $\calF_m$ since $P_\nu(z)$ is a symmetric polynomial. 
We conclude that the polynomial $Q(z)$ defined in Eq.~(\ref{Q(z)}) describes
a fermionic $m$-particle (unnormalized) state 
\be
|Q\ra=\sqrt{m!} P_\nu |\tau\ra.
\ee
The coefficient of the monomial $z(\lambda)$ in $Q(z)$ coincides with 
$\la D_{\lambda+\tau}|Q\ra/m!$, where 
\[
\lambda+\tau = (\lambda_1+\tau_1,\ldots,\lambda_m+\tau_m).
\]
Indeed,  $Q(z)$ is anti-symmetric under the exchange of any pair of variables
while  $D_{\lambda+\tau}(z)/m!$ is the anti-symmetrized version of the monomial $z(\lambda)$
defined in Eq.~(\ref{z(lambda)}).
Thus, the Frobenius character formula can be expressed as
\be
\label{Frob1}
\chi_\lambda(g) = \la D_{\lambda+\tau}|Q\ra/m!=
\la \lambda+\tau | P_\nu |\tau\ra.
\ee

The next step is to rewrite Eq.~(\ref{Frob1}) using the fermionic second quantization formalism.
To this end consider an infinite chain of qubits labeled by integers $k\in \Zbb$. Let $X_k,Y_k,Z_k$
be the Pauli operators acting on the $k$-th qubit. Define a fermionic annihilation operator
\be
\label{JordanWigner}
a_k =\frac12 (X_k+iY_k) \prod_{j<k}  Z_{j}
\ee
and let $a_k^\dag$ be the corresponding creation operator.
Here we used the Jordan-Wigner fermion-to-qubit mapping.
The operators $a_k,a_k^\dag$ obey the Fermi-Dirac commutation rules, that is, $a_i a_j =-a_j a_i$ and
$a_i a_j^\dag + a_j^\dag a_i = \delta_{i,j}I$ for all $i,j$.
Let $\calH$ be the Hilbert space of the chain. It has  basis vectors $|x\ra$,
where  $x\in \{0,1\}^{\Zbb}$ is a bit string with a finite support.

Given a finite strictly decreasing integer sequence $\kappa$, let $x_\kappa\in  \{0,1\}^{\Zbb}$
be a bit string with the support $\kappa$.
The Slater determinant basis state $|\kappa\ra=D_\kappa(z)/\sqrt{m!}$
in the first quantization is represented by the state $|x_\kappa\ra\in \calH$
in the second quantization. 

Multiplication of a  wave function by the polynomial $\sum_{j=1}^m z_j^\ell$
in the first quantization picture shifts the momentum of each particle by $\ell$.
Such shift is described by  an operator 
\be
\label{Jell}
J_\ell = \sum_{k\in \Zbb}  a_{k+\ell}^\dag a_k
\ee
in the second quantization picture. 
Following~\cite{crichigno2024quantumspinchainssymmetric}, we shall refer to $J_\ell$ as the current operator. 
Now the  Frobenius character formula can be expressed as
%DG: fixed notation x(\tau) --> x_{\tau}
\be
\label{Frob2}
\chi_\lambda(g) = \la x_{\lambda+\tau} |\prod_{j=1}^{c(g)} J_{\nu_j}  |x_\tau\ra.
\ee
The order in the product does not matter since the current operators pairwise commute.
At this point we have mapped $S_n$ characters to an infinite spin chain.
We can truncate the chain to a block of $2n$ consecutive qubits using the following lemma.
\begin{lemma}
\label{lemma:truncate}
Suppose $x,y\in \{0,1\}^{\Zbb}$ are bit strings and $i^*$ is an integer 
such that $x_i=y_i$ for $i\le i^*$.
Let $x',y'$  be bit strings obtained from $x,y$ respectively
 by setting to zero all bits $i\le i^*$.  Then 
\be
\label{truncate}
\la y |\prod_{j=1}^{c(g)} J_{\nu_j} |x\ra = \la y'| \prod_{j=1}^{c(g)} J_{\nu_j} | x'\ra.
\ee
\end{lemma}
\begin{proof}
Let  $\tilde{J}_\ell = \sum_{k>i^*}  a_{k+\ell}^\dag a_k$.
One can easily check that 
\[
\prod_{j=1}^{c(g)} J_{\nu_j} |x\ra =\prod_{j=1}^{c(g)} \tilde{J}_{\nu_j} |x\ra + |\phi\ra,
\]
where $|\phi\ra\in \calH$ is a superposition of basis vectors $|z\ra$
such that
$z_i=0$ and $x_i=1$ for at least one qubit $i\le i^*$ (in the case that no sites $i \le i^*$ are occupied, $|\phi\ra$ is $0$). By assumption,
$x_i=y_i$ for all $i\le i^*$.
Thus $|\phi\ra$ is orthogonal to $|y\ra$, that is,
\be
\label{truncate1}
\la y |\prod_{j=1}^{c(g)} J_{\nu_j} |x\ra= 
\la y |\prod_{j=1}^{c(g)} \tilde{J}_{\nu_j} |x\ra.
\ee
Since $\tilde{J}_{\nu_j}$ acts trivially on qubits $i\le i^*$, we also have
\be
\label{truncate2}
\la y |\prod_{j=1}^{c(g)} \tilde{J}_{\nu_j} |x\ra = \la y' |\prod_{j=1}^{c(g)} \tilde{J}_{\nu_j} |x'\ra.
\ee
Applying Eq.~(\ref{truncate1}) again with $x',y'$ instead of $x,y$ proves Eq.~(\ref{truncate}).
\end{proof}
As a consequence of Lemma~\ref{lemma:truncate}, the matrix element in Eq.~(\ref{Frob2}) does not change 
if we set $m=n$ and pad $\lambda$ with $n-m$ zeros on the right. Accordingly, we can set
\[
\tau=(n-1,n-2,\ldots,1,0).
\]

It is immediate that any qubit $i<0$ or $i\geq 2n$ is set to $|0\ra$ in the initial state $|x_\tau \ra$. Note that $\prod_{j=1}^{c(g)}J_{\nu_j}$ can only move a fermion at most $n$ sites to the right. Thus, qubits $i<0$ and $i\ge 2n$ are unchanged in the simulation and can be removed. This leaves only $2n$ active qubits $i=0,1,\ldots,2n-1$.
% Any qubit $i<0$  or $i\ge 2n$ is set to $|0\ra$ in the initial state $|x_\tau\ra$ and remains in the state $|0\ra$ throughout
% application of the current operators (here we noted that an operator $J_{\nu_j}$ can shift the support of a bit string
% to the right at most by $\nu_j$ and $\sum_{j=1}^{c(g)} \nu_j=n$).
% Thus qubits $i<0$ and $i\ge 2n$ can be removed from the simulation. This leaves only $2n$ active qubits $i=0,1,\ldots,2n-1$. 
%LS: I slightly rewrote the first half of the above paragraph in the pursuit of clarity. The oroginal text is given in the comment above.
%DG: added the footnote describing the interpretation of the encoding of partitions
After this truncation the initial basis state $|x_\tau\ra$ is mapped to $|1^n0^n\ra$
while the final basis state $x_{\lambda+\tau}$ is mapped to a weight-$n$  bit string with the support $\{\lambda_1+n-1,\lambda_2+n-2,\ldots,0\}$ \footnote{It may be helpful to understand this binary encoding of a partition as follows.  If we start with the bit string $x_{\tau+\lambda}$ and then interpret each $1$ as ``up by one'' and each $0$ as ``to the right by one'' and then trace out a picture according to these instructions, we get a picture of the partition $\lambda$ within a bounding box of size $n\times n$.}.
Now the  Frobenius character formula becomes
\be
\label{Frob3}
\chi_\lambda(g) =
\la x_{\lambda+\tau} | \prod_{j=1}^{c(g)} J_{\nu_j}  |1^n 0^n\ra, 
\ee
where
\be
J_\ell = \sum_{k=0}^{2n-1-\ell} a_{k+\ell}^\dag a_k
\ee
This concludes the derivation of the Crichigno and Prakash mapping from characters of $S_n$ to
a quantum spin chain.

As pointed out in~\cite{crichigno2024quantumspinchainssymmetric}, the state $|\psi_g\ra =  \prod_{j=1}^{c(g)} J_{\nu_j}  |1^n 0^n\ra$ has 
support only on bit strings encoding partitions of $n$.
This can be  seen as follows. Let $x=(x_0,x_1,\ldots,x_{2n-1})$ be a bit string of length $2n$.
Let  $|x|=\sum_{j=0}^{2n-1} x_j$ be the weight of $x$
and  $m(x)=\sum_{j=0}^{2n-1} j x_j$ be the momentum of $x$.
For example,  $m(1^n0^n) = \sum_{j=0}^{n-1} j = n(n-1)/2$. We claim that $x$
encodes a partition of $n$ if and only if $x$ has weight $n$ and momentum $n(n+1)/2$.
Indeed, suppose $x=x_{\lambda+\tau}$ for some partition $\lambda \vdash n$.
Then $x$ has weight $n$ since we agreed that $\lambda$ has $n$ parts (including the padded zero parts).
By definition, $x$ has support $\{ \lambda_j+n-j\}_{j=1,\ldots,n}$.  From  $\sum_{j=1}^n \lambda_j=n$ one gets
\[
m(x) = \sum_{j=1}^{n}(\lambda_j+n-j) = n(n+1)/2.
\]
Conversely, suppose $x$ has weight $n$ and momentum $n(n+1)/2$.
Let $2n-1\ge d_1>d_2>\ldots>d_n\ge 0$ be the indices of non-zero bits of $x$.
Set $\lambda_j = j+d_j-n$ for $j=1,\ldots,n$.
From $d_j>d_{j-1}$ one gets  $n\ge \lambda_1\ge \ldots\ge \lambda_n\ge 0$.
Also we have $\sum_{j=1}^n d_j = m(x)$, that is,
\[
\sum_{j=1}^n \lambda_j = m(x)-n(n-1)/2 = n.
\]
Thus $\lambda$ is a partition of $n$.

Clearly, the current operators $J_{\nu_j}$ preserve the weight of bit strings. 
Each application of $J_{\nu_j}$ increases the momentum by $\nu_j$.
Since the $\sum_{j=1}^{c(g)} \nu_j= n$ and the initial bit string $1^n0^n$ has momentum $n(n-1)/2$,
the final state is a superposition of bit strings with  weight $n$ and momentum $n(n-1)/2 + n =n(n+1)/2$.
As shown above, these are exactly the bit strings encoding partitions of $n$.

\section{Classical MPS  algorithms}
\label{sec:classical}

%SBB: fixed an inconsistent notation for bitstrings x(lambda)
Consider a permutation $g\in S_n$. Let $c(g)$ be the number of cycles in $g$ and  $\nu_j$ be the length of the $j$-th cycle.
Our strategy is to represent the $2n$-qubit state 
\be
\label{psi_g}
|\psi_g\ra = \prod_{j=1}^{c(g)} J_{\nu_j}  |1^n0^n\ra
\ee
by a Matrix Product State (MPS). To this end we need the following lemma.
\begin{lemma}
\label{lemma:mpo}
The current operator $J_\ell$ is a Matrix Product Operator (MPO) with the bond dimension $\ell+2$.
\end{lemma}

\begin{proof}
Fig.~\ref{fig:mpo} defines 
a tensor $T\in \Cbb^2 \otimes \Cbb^{\ell+2}\otimes \Cbb^{\ell+2} \otimes \Cbb^2$, where
 $\Cbb^2$ represent physical indices and
$\Cbb^{\ell+2}$ represent virtual indices of the MPO.
\begin{figure}[h]
  \includegraphics[height=2cm]{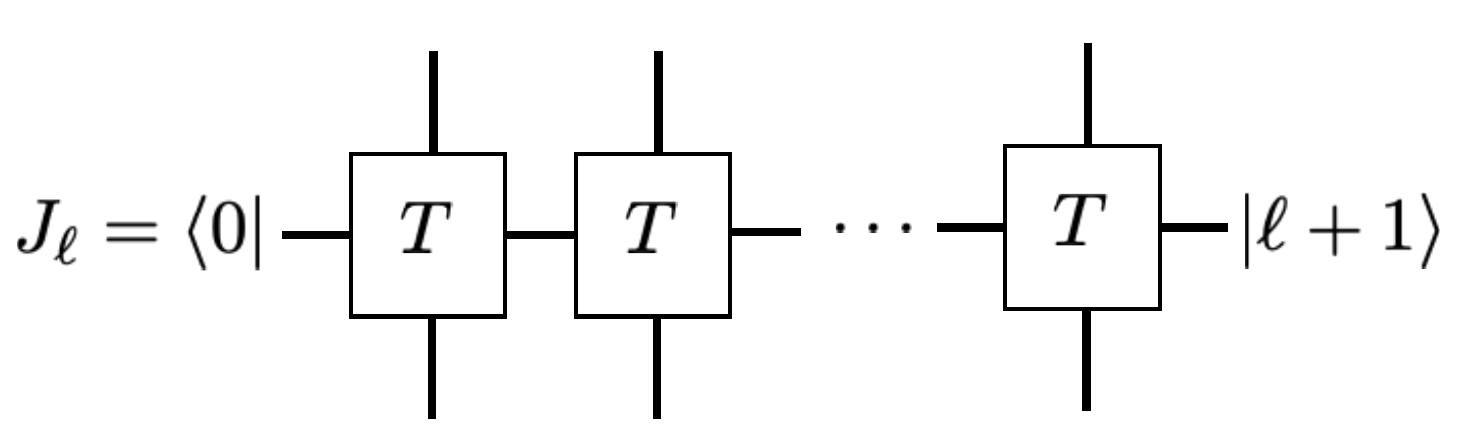}
         \includegraphics[height=4cm]{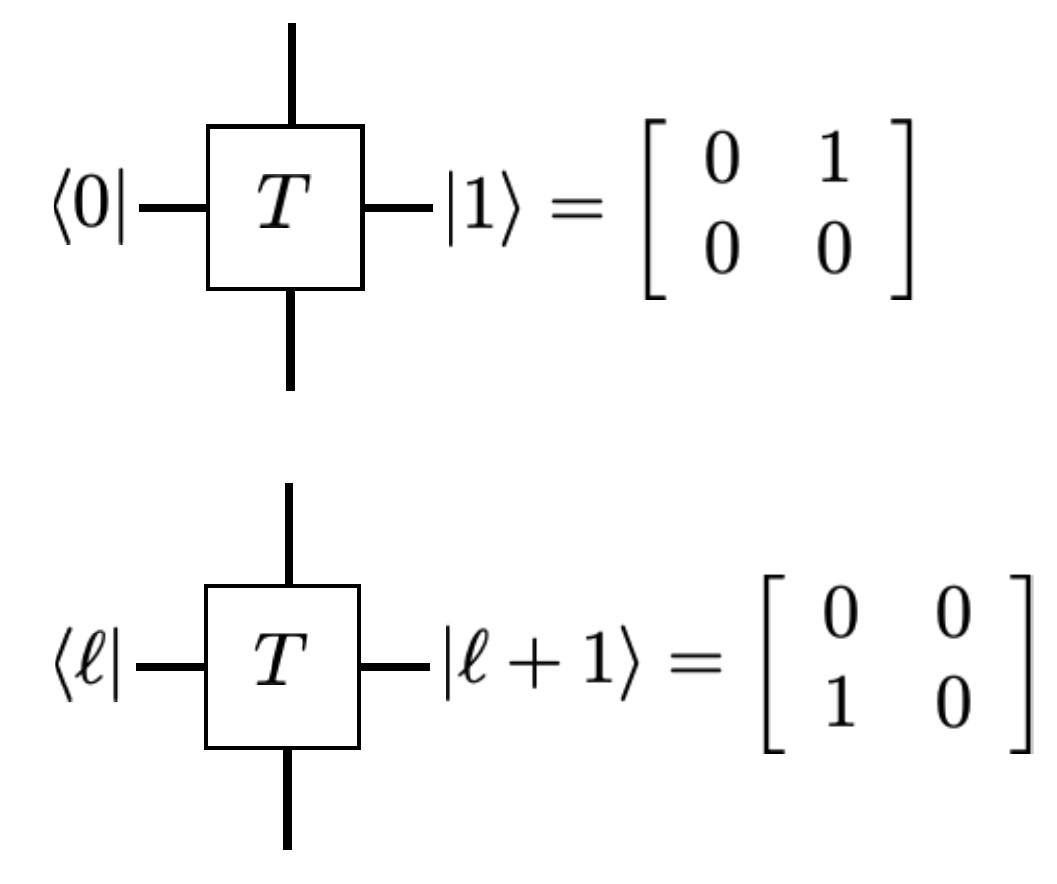}
          \includegraphics[height=4cm]{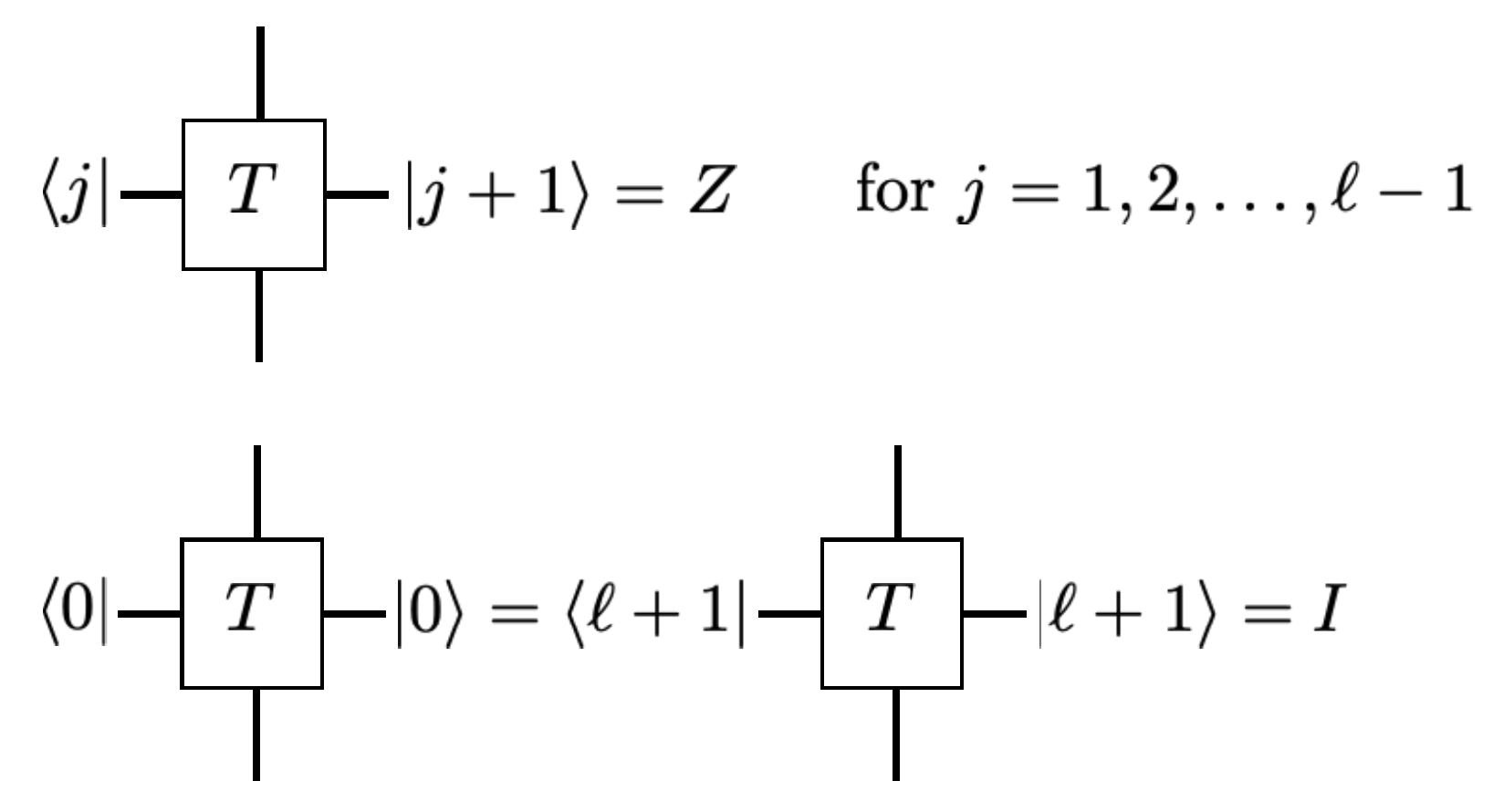}
     \caption{MPO representation of the current operator $J_\ell$. Horizontal lines
    represent virtual indices of dimension $\ell+2$. We label basis vectors of $\Cbb^{\ell+2}$ by 
     integers $\{0,1,\ldots,\ell+1\}$. The MPO consists of $2n$ copies of $T$. We list all nonzero components of the tensor $T$.}
         \label{fig:mpo}
 \end{figure}
 Taking $2n$ copies of $T$ and contracting all virtual indices
(horizontal lines) as shown on Fig.~\ref{fig:mpo} gives an operator
\[
\sum_{j=0}^{2n-1-\ell} 
\left[\begin{array}{cc}
0 & 1 \\
0 & 0 \\
\end{array}\right]_j Z_{j+1} \cdots Z_{j+\ell-1}  \left[\begin{array}{cc}
0 & 0 \\
1 & 0 \\
\end{array}\right]_{j+\ell}.
\]
This is the current operator defined in Eqs.~(\ref{JordanWigner},\ref{Jell}).
\end{proof}
The initial state $|1^n0^n\ra$ is clearly an MPS with bond dimension $1$.
Define a sequence of Matrix Product States
%SBB1: fixed a typo in the index of phi_j states
$|\phi_0\ra = |1^n0^n\ra$ and
\be
|\phi_{j}\ra = \mathrm{Compress}(J_{\nu_j}|\phi_{j-1}\ra, \epsilon)
\ee
for $j=1,\ldots,c(g)$. $\mathrm{Compress}(|\phi\ra,\epsilon)$ is a function that
computes the canonical form of an MPS $|\phi\ra$ 
 and truncates the singular values (Schmidt coefficients) at each site
as discussed in Section~4.5.1 of~\cite{schollwock2011density}.
The error tolerance $\epsilon\ge 0$ defines the  maximum relative error incurred while
truncating singular values. More precisely, suppose $s_1\ge s_2\ge \ldots \ge s_N$ are 
singular values of some MPS $|\phi_j\ra$ for some bipartite cut of the chain. We discard
singular values $s_i$ with $i\ge d$ where $d$ is the smallest index  such that 
\[
\sum_{i=d+1}^N s_i \le \epsilon \sum_{i=1}^N s_i.
\]
We chose to impose a bound on the relative truncation error since the considered states are unnormalized
and some singular values $s_i$ can be exponentially large. 
Let $D=D(\epsilon)$ be the maximum bond dimension of all states $|\phi_j\ra$ with
$j\le c(g)$.  Lemma~\ref{lemma:mpo} ensures that 
\be
\label{bond_dimension}
D \le  \prod_{j=1}^{c(g)} (\nu_j+2).
\ee
We do not expect this bound to be tight, even for lossless compression, $\epsilon=0$.
 For example, one can easily check that a state
$J_{\nu_1}|1^n0^n\ra$ has bond dimension only $\nu_1$ instead of $\nu_1+2$ predicted by Eq.~(\ref{bond_dimension}).
We empirically observed that the MPS bond dimension $D$ at the compression error $\epsilon=10^{-10}$  can be smaller
than the upper bound of Eq.~(\ref{bond_dimension}) by several orders of magnitude, see Fig.~\ref{fig:bond_dimension}.

Since the cost of MPO-MPS multiplication and MPS compression scales as $O(nD^3)$,
an MPS representation of $|\psi_g\ra$ can be computed in time $O(nc(g)D^3)=O(n^2 D^3)$.
Let 
\[
|\mathsf{mps}_g\ra = |\phi_{c(g)}\ra
\]
be the final MPS approximating $|\psi_g\ra$.
The character of any irrep $\lambda\vdash n$ then can be extracted by computing a single amplitude of the final MPS,
\be
\chi_\lambda(g) = \la x_{\lambda+\tau}|\psi_g\ra \approx  \la x_{\lambda+\tau}|\mathsf{mps}_{g}\ra.
\ee
This takes time $O(nD^3)$. 
%SBB: new paragraph
Suppose one needs to compute 
MPS amplitudes $\la x_{\lambda+\tau}|\psi_g\ra$ for all irreps $\lambda$.
To speedup this computation, we divided the spin chain into four 
blocks of length $n/2$ and cached products of the MPS tensors associated with each block. As one iterates over irreps $\lambda$, a product of the MPS tensors is retrieved from the cache, if available. We observed that the runtime for computing MPS amplitudes 
with the caching strategy described above
is comparable or smaller than the runtime for computing the MPS tensors.

Recall that $S_n$ characters are integer-valued. Thus the exact value of $\chi_\lambda(g)$ can be 
recovered by rounding the final MPS amplitude to the nearest integer as long as 
\be
\label{desired_approximation}
| \chi_\lambda(g)  -  \la x_{\lambda+\tau}|\mathsf{mps}_{g}\ra |<\frac12.
\ee
An error introduced by the compression can be estimated by comparing the norm of the exact state $|\psi_g\ra$ and the approximating
MPS. Orthogonality theorem for characters~\cite{etingof2011introduction} gives
\be
\|\psi_g\|^2 = \sum_{\lambda \vdash n} \chi_\lambda(g)^2 = |E_g|,
\ee
where $E_g = \{h\in S_n\, : \, hg=gh\}$ is the centralizer of $g$.
Denoting  the number of length-$\ell$ cycles in $g$ by $a_\ell$ one gets
\be
\label{centralizer_size}
|E_g|=\prod_{\ell=1}^n (a_\ell !) \ell^{a_\ell}.
\ee
%SBB: many changes in the rest of this section
We used $\mathsf{mpnum}$ software~\cite{suess2017mpnum} for numerical experiments
and set the compression error to  $\epsilon=10^{-10}$.
Runtime of the MPS algorithm is compared with the runtime of two widely used
computer algebra systems: $\mathsf{GAP}$~\cite{GAP4} and $\mathsf{symmetrica}$
library included in $\mathsf{SAGE}$~\cite{SAGE}, see Fig.~\ref{fig:runtime}.

\begin{figure}[h]
  \includegraphics[height=6cm]{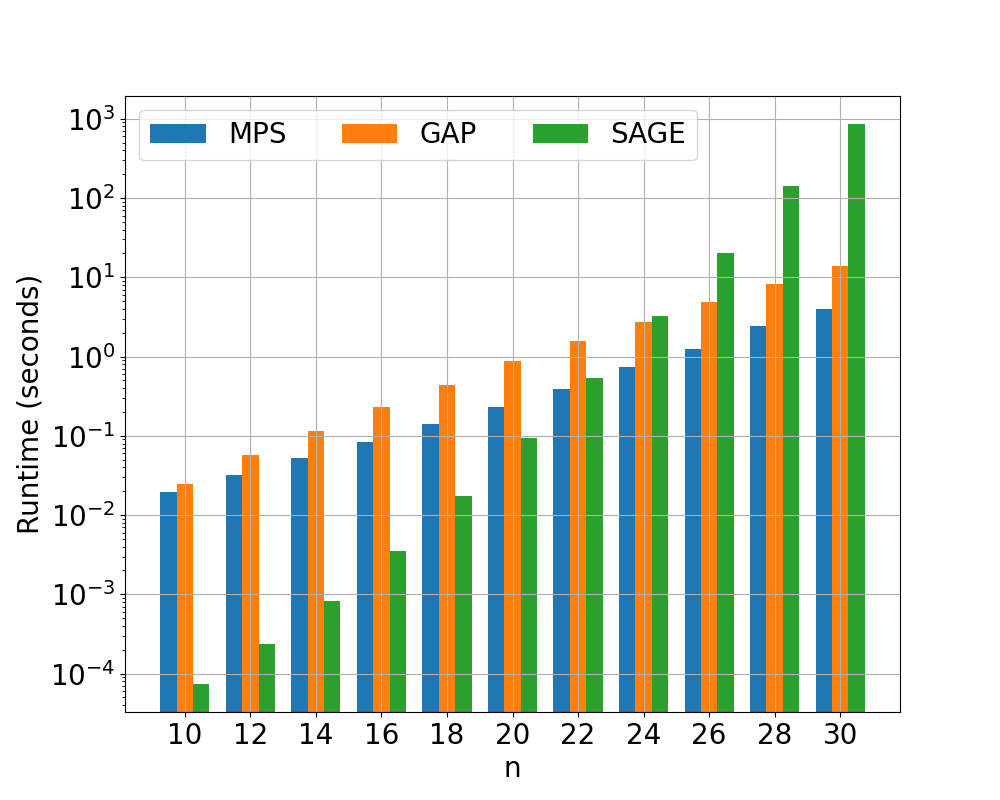}
     \caption{Runtime  for computing the characters $\chi_\lambda(g)$ for permutations $g\in S_n$
     composed of $n/2$ length-$2$ cycles and all irreps $\lambda\vdash n$. 
     The MPS compression error is set to $\epsilon=10^{-10}$ for all the data.        }
         \label{fig:runtime}
 \end{figure}

\begin{figure}[h]
  \includegraphics[height=6cm]{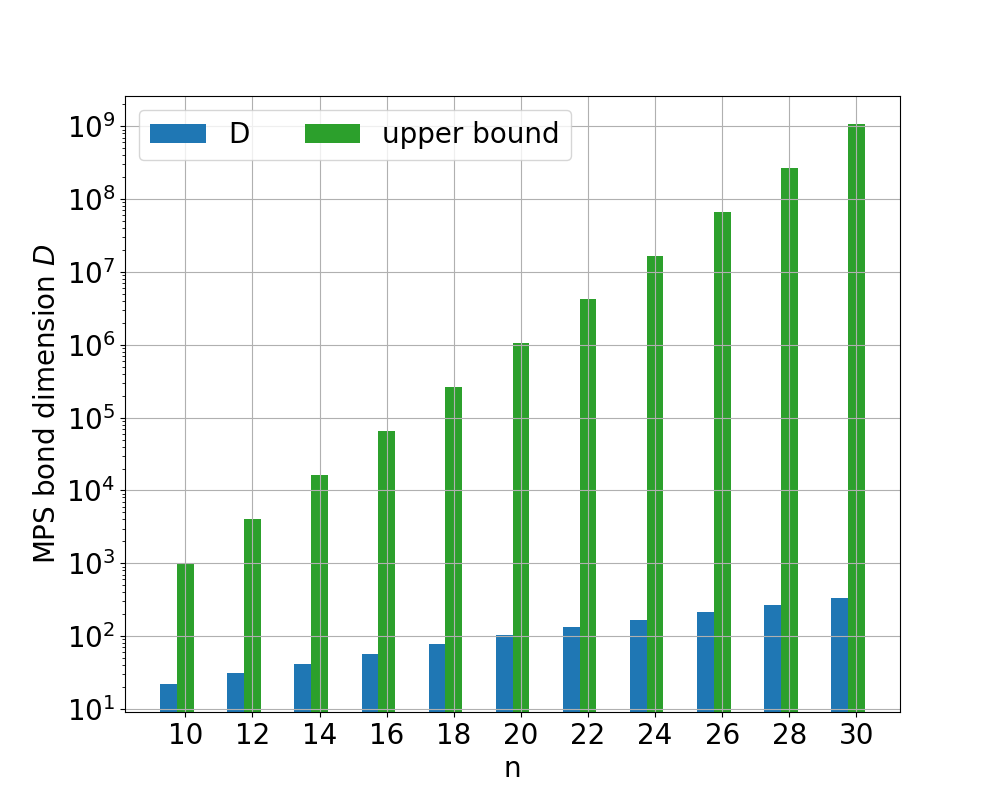}
     \caption{Blue: maximum MPS bond dimension $D$ for solving each problem instance reported on Fig.~\ref{fig:runtime}.
     Green: upper bound of Eq.~(\ref{bond_dimension}). 
       }
         \label{fig:bond_dimension}
 \end{figure}

For numerical  experiments we selected permutations with a large number of short cycles which tend to be most expensive for all considered algorithms. 
Namely, for each $n$ we consider a single permutation $g\in S_n$ composed  of $n/2$ length-$2$ cycles.
Fig.~\ref{fig:runtime} shows the measured runtime  for 
 $10\le n\le 30$. The reported runtimes include computation of the characters $\chi_\lambda(g)$ for all irreps $\lambda \vdash n$.
In the case of MPS the runtime additionally includes computing the state $|\mathsf{mps}_{g}\ra$. As one can see from the plot, the MPS algorithm 
shows a competitive performance for large values of $n$. We note that
our implementation is in Python 
while $\mathsf{SAGE}$ and $\mathsf{GAP}$ have a kernel written in C, which tends to be much faster than Python. 

Recall that our MPS algorithm computes the character $\chi_\lambda(g)$ with the zero error as long as the MPS
amplitude $\la x_{\lambda+\tau}|\mathsf{mps}_{g}\ra$ obeys Eq.~(\ref{desired_approximation}).
We found that
\[
|\chi_\lambda(g) - \la x_{\lambda+\tau}|\mathsf{mps}_{g}\ra|\le 10^{-6}
\]
for all problem instances reported at Fig.~\ref{fig:runtime} (the exact value of $\chi_\lambda(g)$ was computed using $\mathsf{SAGE}$ library).
 This is many orders of magnitude smaller than the required precision $1/2$ suggesting that 
the MPS runtime may be improved further by truncating more singular values.
Our Python implementation of the MPS algorithm can be found at~\cite{character_builder}.
We provide two versions of the algorithm that use  $\mathsf{mpnum}$ or $\mathsf{quimb}$ libraries
for MPO-MPS manipulations.

Finally, for each problem instance presented on Fig.~\ref{fig:runtime} we computed the maximum MPS bond dimension 
and compared it with the upper bound Eq.~(\ref{bond_dimension}). This comparison is reported  on Fig.~\ref{fig:bond_dimension}.

\section{Quantum algorithm}
\label{sec:quantum}

Here we describe a quantum circuit of size $poly(n)$ that prepares a state $|\psi_g\ra/\|\psi_g\|$,
see Eq.~(\ref{psi_g}).
Our circuit applies the Quantum Fourier Transform over the group $S_n$
to a certain subset state constructed  in the following lemma.
\begin{lemma}
\label{lemma:subset_state}
Let $g \in S_n$ be a permutation and $C_g\subseteq S_n$ be the conjugacy class of $g$.
There exists a quantum circuit of size $poly(n)$ that prepares a subset state
\be
|C_g\ra = \frac1{\sqrt{|C_g|}} \sum_{h \in C_g} |h\ra.
\ee
\end{lemma}
We defer the proof till the end of this section.
Recall that a Quantum Fourier Transform (QFT) over the symmetric group $S_n$ is defined as a unitary operator
$\mathrm{QFT}_n$ with the input Hilbert space $\Cbb[S_n]=\mathrm{span}(|g\ra\, : \, g\in S_n)$ and
output Hilbert space $\bigoplus_{\lambda\vdash n} \Cbb^{d_\lambda}\otimes \Cbb^{d_\lambda}$ such that
\be
\label{eq4}
\mathrm{QFT}_n |g\ra=
 \sum_{\lambda \vdash n} \sqrt{\frac{d_\lambda}{n!}} \sum_{i,j=1}^{d_\lambda} \la i |\rho_\lambda(g)|j\ra \cdot |\lambda,i,j\ra
\ee
for all permutations $g \in S_n$.
%SBB1: new sentence
Here $d_\lambda$ is the dimension of $\lambda$.
The inverse QFT acts as
\be
\mathrm{QFT}_n^\dag |\lambda,i,j\ra = \sqrt{\frac{d_\lambda}{n!}} \sum_{g \in S_n} 
 \la i |\rho_\lambda(g)|j\ra^*
 |g\ra.
\ee
It is known~\cite{beals1997quantum} that $\mathrm{QFT}_n$ can be implemented by a quantum circuit of size $poly(n)$.
Applying $\mathrm{QFT}_n$ to the subset state $|C_g\ra$ 
gives a state
\begin{align}
\label{psi_8}
|\theta_g\ra & = \mathrm{QFT}_n |C_g\ra   \nonumber \\
&=\sum_{\lambda \vdash n}
\sqrt{\frac{d_\lambda}{n!}} \sum_{i,j=1}^{d_\lambda} \frac1{\sqrt{|C_g|}} \sum_{h \in C_g} \la i|\rho_\lambda(h)|j\ra  |\lambda,i,j\ra.
\end{align}
Since $\sum_{h \in C_g} h$ is a central element of the group algebra $\Cbb[S_n]$, we infer that 
\be
\label{central_element}
\sum_{h \in C_g} \rho_{\lambda}(h) = b I_{d_\lambda}
\ee
for some complex number  $b\in \Cbb$. Here $I_{d_\lambda}$ is the identity matrix of size $d_\lambda$.
Taking the trace of Eq.~(\ref{central_element}) gives
\be
\label{normalization}
b = \frac{|C_g| \chi_\lambda(g)}{d_\lambda}.
\ee
Substituting Eqs.~(\ref{central_element},\ref{normalization}) into Eq.~(\ref{psi_8})
and using the identity $|\lambda,i,j\ra =|\lambda\ra \otimes |i,j\ra$
shows that 
\[
|\theta_g\ra =  \sqrt{\frac{|C_g|}{n!}} \sum_{\lambda \vdash n} \chi_\lambda(g) |\lambda\ra \otimes |\phi_\lambda\ra,
\]
where
\[
|\phi_\lambda\ra =  \frac1{\sqrt{d_\lambda}} \sum_{i=1}^{d_\lambda} |i,i\ra.
\]
Let $\calU$ be a unitary operator that maps $|\lambda\ra \otimes |0\ra$ to $|\lambda\ra \otimes |\phi_\lambda\ra$.
%SBB1: added a comment
One can realize $\calU$ by a quantum circuit of size $poly(n)$, for example using the techniques of~\cite{shukla2024efficient},
since the dimension $d_\lambda$
can be computed in time $poly(n)$ using the well-known hook length formula~\cite{stanley1999enumerative}.

%DG: changed \theta to \theta_g below
Applying $\calU^{-1}$ to the state $|\theta_g\ra$ and tracing out the second register gives a state
\[
|\omega_g\ra =  \sqrt{\frac{|C_g|}{n!}} \sum_{\lambda \vdash n} \chi_\lambda(g) |\lambda\ra
=\frac1{\sqrt{|E_g|}} \sum_{\lambda \vdash n} \chi_\lambda(g) |\lambda\ra,
\]
where $E_g$ is the centralizer of $g$, see Eq.~(\ref{centralizer_size}).
This a normalized version of the desired state $|\psi_g\ra$ if one encodes a partition $\lambda\vdash n$
by a bit string $x_{\lambda+\tau}$ as discussed in Section~\ref{sec:spin_chain}.
By measuring the state $|\omega_g\ra$ in the standard basis one can sample an irrep $\lambda\vdash n$ from
a probability distribution 
\begin{equation}
P_g(\lambda) = \frac{(\chi_\lambda(g))^2}{|E_g|}.   \qquad (\textbf{Row sampling})
\label{eq:rowsamp}
\end{equation}
This is equivalent to sampling a row of the character table of $S_n$ with the probability proportional to the character squared.

A similar method can prepare a superposition of all permutations $g\in S_n$ with amplitudes
proportional to $\chi_\lambda(g)$ for a given irrep $\lambda\vdash n$. Indeed,
apply $\mathrm{QFT}_n^\dag$ to an initial state
\[
|\theta_\lambda\ra = \frac1{\sqrt{d_\lambda}} \sum_{i=1}^{d_\lambda} |\lambda,i,i\ra.
\]
This gives a state
\[
|\omega_\lambda\ra = \mathrm{QFT}_n^\dag|\theta_\lambda\ra = \frac1{\sqrt{n!}} \sum_{g \in S_n} \chi_\lambda(g)^* |g\ra
\]
(the complex conjugate can be omitted  since $S_n$ characters are real).
By measuring the state $|\omega_\lambda\ra$ in the standard basis one can sample a permutation $g\in S_n$ from
a probability distribution 
\begin{equation}
P_\lambda(g) = \frac{(\chi_\lambda(g))^2}{n!}.  \qquad (\textbf{Column sampling})
\label{eq:colsamp}
\end{equation}
This is equivalent to sampling a column of the character table of $S_n$ with the probability proportional to the character squared
(here we assume that columns are labelled by permutations rather than conjugacy classes).
%SBB1: new sentence
We discuss the classical hardness of these sampling problems in the next section.

\begin{proof}[\bf Proof of Lemma~\ref{lemma:subset_state}]
Let $k=|C_g|$ and 
\[
C_g = \{h_1,h_2,\ldots,h_k\}.
\]
Suppose one can construct permutations  $f_1,f_2,\ldots,f_k\in S_n$  such that 
\be
\label{mu's}
h_j = f_j g f_j^{-1}
\ee
for all $j=1,\ldots,k$ and
there exists a quantum circuit $\calQ$ of size $poly(n)$
such that 
\be
\label{calQ}
\calQ|h_j\ra \otimes |0\ra = |h_j\ra \otimes |f_j\ra
\ee
for all $j=1,\ldots,k$.
Then the desired state $|C_g\ra$ can be prepared as follows.
First prepare a state 
\[
|\psi_1\ra = \frac1{\sqrt{n!}} \sum_{e \in S_n} |e\ra \otimes |g\ra.
\]
Apply a classical reversible circuit that maps $|e\ra\otimes |g\ra$
to $|e\ra \otimes |e g e^{-1}\ra$. This gives a state
\[
|\psi_2\ra = \frac1{\sqrt{n!}} \sum_{e \in S_n} |e\ra \otimes |e g e^{-1}\ra.
\]
Let $E_g$ be the centralizer of $g$, 
\[
E_g = \{ b \in S_n \, : \,  g b = b g \}.
\]
Since left cosets of $E_g$ in $S_n$ can be identified with elements of $C_g$,
any permutation $e \in S_n$ can be uniquely decomposed as
$e = f b$ for some $f \in \{f_1,\ldots,f_k\}$ and some
$b \in E_g$. Thus
\[
|\psi_2\ra = \frac1{\sqrt{k |E_g|} } \sum_{j=1}^k \sum_{b \in E_g} |f_j b\ra \otimes |f_j g f_j^{-1}\ra.
\]
Here we noted that $n!=k|E_g|$. Tensoring in the third register
prepared in the $|0\ra$ state 
and applying the circuit Eq.~(\ref{calQ}) to the second and the third register gives a state
\[
|\psi_3\ra = \frac1{\sqrt{k |E_g|} } \sum_{j=1}^k \sum_{b \in E_g} |f_j b\ra \otimes |f_j g f_j^{-1}\ra \otimes |f_j\ra.
\]
Next apply a classical reversible circuit that maps $|f_j b\ra \otimes |f_j\ra$
to $|b\ra \otimes |f_j\ra$ to the first and the third registers of $|\psi_3\ra$. This gives a state
\[
|\psi_4\ra = \frac1{\sqrt{k |E_g|} } \sum_{j=1}^k \sum_{b \in E_g} |b\ra \otimes |f_j g f_j^{-1}\ra \otimes |f_j\ra.
\]
Now the first register is left in a pure state $(1/\sqrt{|E_g|})\sum_{b\in E_g} |b\ra$.
Tracing out the first register gives a state
\[
|\psi_5\ra =  \frac1{\sqrt{k }} \sum_{j=1}^k  |f_j g f_j^{-1}\ra \otimes |f_j\ra.
\]
Applying the inverse of the circuit Eq.~(\ref{calQ}) gives a state
\[
|\psi_6\ra =  \frac1{\sqrt{k }} \sum_{j=1}^k  |f_j g f_j^{-1}\ra \otimes |0\ra.
\]
Finally, tracing out the second register we obtain a state
\[
|\psi_7\ra =   \frac1{\sqrt{k }} \sum_{j=1}^k  |f_j g f_j^{-1}\ra =
\frac1{\sqrt{|C_g|}} \sum_{h \in C_g} |h\ra = |C_g\ra.
\]
Here the second equality makes use of Eq.~(\ref{mu's}) and 
the equality $C_g = \{h_1,h_2,\ldots,h_k\}$.

It remains to construct permutations $f_1,\ldots,f_k$ 
and a quantum circuit $\calQ$ satisfying Eqs.~(\ref{mu's},\ref{calQ}).
Suppose $g$ consists of $a_\ell$ cycles of length $\ell$,
where $a_\ell$ is an integer in the range $[0,n]$ and
\begin{equation}
n=\sum_{\ell=1}^n \ell a_\ell.
\label{eq:nal}
\end{equation}
Clearly, the conjugacy class $C_g$ depends only on $a_1,\ldots,a_n$.
Let $T$ be a Young tableau with $n$ boxes and  $a_\ell$ rows of length $\ell$ 
for $\ell=1,\ldots,n$. We fill $T$ with integers $1,2,\ldots,n$ in the row major order
starting from the leftmost box in the first (top) row.
Now integers $1,2,\ldots,n$ can be identified with boxes of $T$.
We can assume wlog that $g$ cyclically shifts boxes in each row of $T$
such that $i$-th box in a row $j$ goes to $(i+1)$-th box in the same row
and the last box in a row $j$ goes to the first box in the same row.

Consider a permutation $h \in C_g$. Clearly, $h$ has the same cyclic type as $g$
although boxes spanning each cycle of $h$ may be scattered over many rows of $T$.
This determines a partition of $T$ into disjoint union of subsets
\[
T=\bigcup_{\ell=1}^n \bigcup_{j=1}^{a_\ell} T_{\ell,j}
\]
such that $|T_{\ell,j}|=\ell$ and
boxes in each subset $T_{\ell,j}$ span a length-$\ell$ cycle in $h$.
Let $b_{\ell,j} \in T_{\ell,j}$ be a   box contained in $T_{\ell,j}$ with the smallest
integer label
(recall that each box of $T$ is labeled by an integer between $1$ and $n$).
Order the subsets $T_{\ell,j}$ with a fixed $\ell$ such that 
\[
b_{\ell,1}<b_{\ell,2}<\ldots<b_{\ell,a_\ell}.
\]
Let $B_{\ell,j}$ be the first (leftmost) box of $T$ in the $j$-th length-$\ell$ row
(rows are counted starting from the top). 
Let $f\in S_n$ be a permutation such that 
\be
\label{mu_def}
f(g^p(B_{\ell,j})) = h^p(b_{\ell,j}) \quad \mbox{for all} \quad p=0,1,\ldots,\ell-1
\ee
and for all $\ell,j$.
Here $g^p \in S_n$ denotes the $p$-th power of $g$.
Since any box of $T$ can be obtained from some box $B_{\ell,j}$ by applying powers of $g$,
the permutation $f$ is uniquely determined by $h$.
Also it is clear that  the map $h \to f$ is computable in time $poly(n)$. Next we claim that 
 $h = f g f^{-1}$.
Since $h$ cyclically shifts boxes in each subset $T_{\ell,j}$, 
it suffices to check that 
\be
h^p(b_{\ell,j}) = (f g f^{-1})^p (b_{\ell,j})
\ee
for all $\ell,j,p$. The above is equivalent to 
\be
h^p(b_{\ell,j})  = f( g^p f^{-1}(b_{\ell,j})) = f(g^p (B_{\ell,j})),
\ee
which follows from Eq.~(\ref{mu_def}).
A classical circuit of size $poly(n)$ computing the map $h \to f$ can be converted to a
poly-size
quantum circuit $\calQ_1$ 
computing a map
\[
\calQ_1 |h\ra \otimes |0\ra \otimes |0\ra = |h\ra \otimes |\mathrm{garbage}(h)\ra \otimes |f\ra.
\]
Here $\mathrm{garbage}(h)$ represents auxiliary bits used to make all gates reversible.
Making a copy of the third register in the standard basis and applying the inverse of $\calQ_1$ gives
a poly-size circuit $\calQ$ computing a map
\[
\calQ |h\ra \otimes |0\ra = |h\ra \otimes |f\ra,
\]
where $h=fgf^{-1}$.
This is the circuit defined in Eq.~(\ref{calQ}).
%SBB1: I don't think we need the rest of this section
%Finally, apply a quantum circuit implementing a map $|h\ra \otimes |f\ra \to |f^{-1} h f\ra \otimes |f\ra$. By construction,  $h = f g f^{-1}$. We obtain a poly-size quantum circuit $\calQ_3$
%computing a map
%\[
%\calQ_3 |h\ra \otimes |0\ra = |g\ra \otimes |f\ra.
%\]
%The first output register %is in a fixed state $|g\ra$ which does not depend on $h$.
%Ignoring the first output register gives the desired circuit $\calQ$.
\end{proof}

\section{Classical complexity of sampling}
\label{sec:granular}
In Section \ref{sec:quantum} we showed that a quantum computer can efficiently sample from probability distributions Eqs. \eqref{eq:rowsamp} and \eqref{eq:colsamp} associated with a given column or row of the character table of the symmetric group $S_n$. Are these sampling problems hard for classical computers?

\subsection{Granularity and quantum advantage}
Let us first consider a general setting in which a sampling problem is defined by a probability distribution $P$ over some finite set  $\Omega$ called the sample space and an error tolerance $\epsilon$. We  consider infinite families of problem instances labeled by an integer $n$ called the problem size. It is required that each point $x\in \Omega$ can be specified by $poly(n)$ bits. Likewise, the probability distribution $P$ should admit a succinct classical description of size  $poly(n)$. We also assume that the error tolerance obeys $\epsilon \ge poly(1/n)$. A classical or quantum algorithm solves the sampling problem if it outputs a sample $x\in \Omega$ from some probability distribution $P'$
which is $\epsilon$-close to $P$ in the total variation distance, such that  $\|P-P'\|_1\equiv \sum_{x\in \Omega} |P(x)-P'(x)|\le \epsilon$. Such an algorithm is  efficient if its runtime is at most $poly(n)$.

We are specifically interested in sampling problems that can be solved efficiently by a quantum computer. To provide evidence for classical hardness of such a sampling problem one can follow the complexity-theoretic argument pioneered by Refs. \cite{aaronson2011computational, bremner2011classical}. Roughly speaking, this argument gives a reduction from \textit{average case strong simulation}---computing output probabilities of a quantum circuit in the average case over a suitable ensemble---to the sampling problem of interest. Here the reduction uses polynomial time classical resources as well as access to an NP oracle (i.e, it is a $BPP^{NP}$ algorithm). If the average case strong simulation problem is so hard that it cannot be solved using resources within the polynomial hierarchy, then this implies that no efficient classical algorithm can solve the sampling problem of interest.

Here we describe a variant of this argument for sampling problems with a certain granularity structure, described below.

\begin{dfn}[\bf Granularity]
A probability distribution $P$ over a finite set $\Omega$ is 
called granular if there exists a real number $ \gamma \ge 1$ such that
$P(x)$ is 
an integer multiple of $\gamma /|\Omega|$ for all points $x\in \Omega$.
We refer to the number $\gamma $ as the granularity of $P$.
\end{dfn}
For example, the distribution Eq.~(\ref{eq:colsamp}) defined by columns of the character table
has granularity $\gamma=1$ since the sample space has size $|\Omega|=|S_n|=n!$
and characters of the symmetric group are integer-valued. The distribution Eq.~(\ref{eq:rowsamp}) defined by rows of the character
table may or may not be granular depending on the choice of the permutation $g \in S_n$, as detailed below.
If the input of a sampling problem is a granular distribution $P$, the granularity $\gamma $
has to be included in the classical description of $P$.

\begin{lemma}
\label{lemma:david}
Suppose $P_0$ and $P_1$ are probability distributions over a finite set $\Omega$
such that $P_0$ is granular, $P_1$ can be sampled by a BPP algorithm,
and $\|P_0-P_1\|_1\le \epsilon$ for a given $\epsilon$
 which  is at least inverse polynomial in $n$. Let $\gamma \ge 1$ be the granularity of $P_0$.
There exists a $BPP{}^{NP}$ algorithm that takes as input a point
$x\in \Omega$ and outputs an integer $k(x)$ such that
\be
P_0(x) = \frac{ \gamma k(x)}{|\Omega|}
\ee
for at least $\max{\{0,1-20\epsilon\}}$ fraction of points $x\in \Omega$. 
\end{lemma}
\begin{proof}
Let $N=|\Omega|$.
Define sets 
\[
S = \left\{ x\in \Omega\ \, : \, P_0(x)\le \frac{1}{10 N\epsilon}\right\}
\]
and
\[
T = \left\{  x\in \Omega \, : \, |P_0(x)-P_1(x)|\le \frac{1}{10 N}\right\}.
\]
Markov's inequality gives  $|S| \ge (1-10\epsilon)N$
and  $|T|\ge (1-10\epsilon)N$. By the union bound, 
\be
\label{ScapT}
|S\cap T| \ge (1-20\epsilon) N.
\ee
The assumption that $P_1$ can be sampled by a $BPP$ algorithm and 
Stockmeyer's approximate counting theorem \cite{stockmeyer1983complexity} imply that there exists a $BPP{}^{NP}$ algorithm that takes
as input a point $x\in \Omega$ and outputs a real number $P_2(x)$ such that 
\[
|P_2(x)-P_1(x)|\le \epsilon P_1(x).
\]
For any $x\in S\cap T$ one has $P_1(x)\le P_0(x) + 1/(10N)  \le 1/(10N\epsilon) + 1/(10N) \le 1/(5N\epsilon)$ which imples $|P_2(x)-P_1(x)|\le 1/(5N)$.
 Thus for any $x\in S\cap T$ one has
\begin{align}
\label{P_2}
|P_2(x) - P_0(x)| & \le |P_2(x)-P_1(x)| + |P_0(x)-P_1(x)| \nonumber \\
& \le
\frac1{5N} + \frac1{10N} \le \frac1{3N}.
\end{align}
By assumption, $P_0(x)$ is an integer multiple of $\gamma /N$ for a given $\gamma \ge 1$.
Let $P_3(x)$ be an integer multiple of $\gamma/N$ closest to $P_2(x)$. 
Using
Eq.~(\ref{P_2}) one gets $P_3(x)=P_0(x)$
whenever $x\in S\cap T$.
The above shows that $P_3(x)$ can be computed by  a  $BPP{}^{NP}$ algorithm
and Eq.~(\ref{ScapT}) implies that $P_3(x)=P_0(x)$ for at least $1-20\epsilon$ fraction of points $x\in \Omega$.
Finally, set  $k(x)=(N/\gamma )P_3(x)$.
\end{proof}
A central challenge in quantum complexity theory is to prove the quantum advantage conjecture \cite{aaronson2011computational} in the following form.
\begin{conj}
Assume that the Polynomial Hierarchy (PH) does not collapse.
Then there exists a sampling problem  which can be solved efficiently on a quantum computer and which is provably hard for classical randomized algorithms. 
\end{conj}
Lemma~\ref{lemma:david} implies that the conjecture is true
if there exists a sampling problem satisfying the following conditions.\\

\noindent
{\bf (1) Granularity:} the sampling problem is based on a family of granular probability distributions $P$,\\

\noindent
{\bf (2) Quantum algorithm:} the sampling problem can be solved efficiently on a quantum computer,\\

\noindent
{\bf (3) Average-case hardness:}
The problem of computing $P(x)$ exactly on a 
point $x\in \Omega$  picked uniformly at random with a sufficiently high success probability 
is not contained in the Polynomial Hierarchy. Here the success probability must be at least $1-poly(1/n)$ for problem instances of size $n$. 

We see that for granular distributions the hardness of classical sampling follows from the average-case hardness property (3) and the assumption that the polynomial hierarchy does not collapse. This differs from the usual argument in that we do not need to establish the so-called anticoncentration property for the family of distributions $P$; moreover, the average case hardness property (3) concerns \textit{exact} computation of output probabilities rather than approximate computation.

%SBB: added a comment
Note that a dummy algorithm for computing $P(x)$ that outputs zero for all $x\in \Omega$
 has success probability at least $1-1/\gamma$. 
 Indeed, the dummy algorithm succeeds with the probability $|S|/|\Omega|$, where
 $S=\{ x \in  \Omega\, : \, P(x)=0\}$.
Granularity gives $P(x)\ge \gamma/|\Omega|$ for
$x\notin S$.
Thus $1=\sum_{x\in \Omega} P(x) \ge \gamma (|\Omega|-|S|)/|\Omega|$, that is,
$|S|/|\Omega| \ge 1-1/\gamma$, as claimed.  This shows that the average-case hardness property (3) can only be achieved if 
the granuarity $\gamma$ is most polynomial in $n$.

\subsection{Row and column sampling}
We have already seen that the column sampling problem obeys properties (1) and (2).
Furthermore, computing the probability  $P_\lambda (g)=|\chi_\lambda(g)|^2/n!$
is $\#P$-hard in the worst-case~\cite{Ikenmeyer24}. 
%SBB3: new sentence
Moreover, this problem is not contained in $\#P$ unless the Polynomial Hierarchy collapses to the second level~\cite{Ikenmeyer24}.
However it seems unlikely that average-case hardness property (3) holds for this problem. In particular,  the following lemma gives a classical quasi-polynomial time  algorithm for computing $P_\lambda (g)$ that succeeds with high probability if  $g \in S_n$ is a random permutation.

\begin{lemma}
\label{lemma:MNrule}
There exists a classical algorithm that takes as input an irrep $\lambda$ of the symmetric group $S_n$,
a permutation $g\in S_n$, and outputs an integer $k(g)$ such that $k(g)=\chi_\lambda(g)$
for at least $1-1/n^{100}$ fraction of permutations $g\in S_n$. The algorithm has runtime $n^{O(\log{n})}$.
\end{lemma}
\begin{proof}
The proof combines the MPS algorithm of Section \ref{sec:classical} for computing characters of $S_n$ 
with tail bounds for the  number of cycles in a random permutation. We shall use the algorithm from Section \ref{sec:classical} with zero truncation error ($\epsilon=0$). With input $\lambda\vdash n$ and $g\in S_n$, this algorithm computes an MPS $|\psi_g\rangle$ of bond dimension Eq.~\eqref{bond_dimension} which can be upped bounded as $D\leq (n+2)^{c(g)}$ where $c(g)$ is the number of cycles in $g$. Recall that computing this MPS and extracting the amplitude $\chi_{\lambda}(g)$ uses a runtime upper bounded as
\[
O(n^2D^3)=n^{O(c(g))}.
\]
This runtime is exponential in $n$ in the worst case since the number of cycles $c(g)$ can be linear in $n$.
 However we are only interested in the average-case, that is, random permutations. And random permutations are known to have only a few cycles whp.
The following fact is a special case of Theorem 1.7 in~\cite{ford2021cycle}.
\begin{fact}
\label{fact:cycles}
Suppose $g \in S_n$ is picked uniformly at random. Let $c(g)$ be the number of cycles
in $g$. Then for all $\ell \ge 1$ one has 
\be
\mathrm{Pr}[c(g)\ge 1+\ell H_n] \le 2e^{1-(\ell \log{(\ell)} -\ell + 1)H_n},
\ee
where
$H_n = \sum_{j=1}^n \frac1{j} =\log{(n)} +O(1)$
is the $n$-th harmonic number and $e=2.71828\ldots$.
\end{fact}
\begin{corol}
\label{corol:eigen}
There exists a constant $B > 0$ such that the fraction of permutations 
$g\in S_n$ with $c(g) > B \log{(n)}$ cycles is at most $1/n^{100}$
for all large enough n. 
\end{corol}
If $c(g)\le B\log{n}$, use the MPS algorithm to compute $\chi_\lambda(g)$, which takes time $n^{O(\log{n})}$.
Otherwise, if $c(g)>B\log{n}$, output some fixed integer (say, zero). This gives the desired algorithm.
\end{proof}

Next let us consider the row sampling problem.
Lemma~\ref{lemma:MNrule} suggests that classically hard instances of the problem
can be constructed by choosing the input permutation $g\in S_n$ with a large number of cycles,
for example,  $c(g)\sim \sqrt{n}$. In this case best known classical algorithms for computing the
characters $\chi_\lambda(g)$ appear to require an exponential time.
However, we need to check whether permutations $g$ with a large number of cycles
give rise to instances of the row sampling problem with a granular distribution.
This is addressed in the following lemma.
\begin{lemma}[\bf Granularity for row sampling]
\label{lemma:row_granularity}
There exist a universal constant $A>0$ such that 
the probability distribution $P_g(\lambda)$ over partitions $\lambda \vdash n$ from Eq.~\eqref{eq:rowsamp} is granular whenever
$g\in S_n$ is a permutation with at most $A\sqrt{n}/\log{n}$ cycles.
\end{lemma}
\begin{proof}
Let $\Omega_n$ be the set of partitions $\lambda \vdash n$ (our sample space).
We can write
\be
P_g(\lambda) = \gamma \cdot  \frac{|\chi_\lambda(g)|^2}{|\Omega_n|} \qquad \gamma=\frac{|\Omega_n|}{|E_g|}.
\ee
where $E_g$ is the centralizer of $g$. Since $\chi_\lambda(g)$ is integer for all $\lambda\in \Omega_n$, 
the distribution $P_g(\lambda) $ is granular whenever $\gamma \ge 1$.

Suppose $g$ consists of $a_\ell$ cycles of length $\ell$. Then $|E_g|$ is given by Eq.~(\ref{centralizer_size}). Also, it is known \cite{hardyramanujan} that  $|\Omega_n|\ge e^{A \sqrt{n}}$ for some universal constant $A>0$ (one can choose $A=\pi \sqrt{2/3} - \epsilon$ for any constant $\epsilon>0$)
and all sufficiently large $n$. Thus 
\be
\gamma > \frac{e^{A\sqrt{n}}}{\prod_{\ell=1}^n \ell^{a_\ell} a_\ell !}.
\ee
Since $a_\ell ! \le (a_\ell)^{a_\ell}$, one gets
\be
\gamma \ge \exp{\left[ A\sqrt{n} - \sum_{\ell=1}^n a_\ell \log{(\ell a_\ell)} \right]}.
\ee
From Eq.~(\ref{eq:nal}) one gets $\ell a_\ell \le n$. Using this along with $c(g)=\sum_{\ell=1}^{n} a_\ell$  gives
\be
\gamma \ge \exp{\left[ A\sqrt{n} - c(g) \log{n} \right]}.
\ee
Thus $\gamma \ge 1$ whenever $A\sqrt{n} - c(g) \log{n}\ge 0$.
\end{proof}
We conclude that the row sampling problem 
for permutations $g$ with $A\sqrt{n}/\log{n}$ or less cycles
obeys properties (1,2). Proving or disproving property (3)--- average-case hardness of computing
$P_g(\lambda)$ for random partitions $\lambda \vdash n$---is left as a challenge for future work. 

\section{Classical MPS Algorithm for Kostka Numbers}
\label{sec:Kostka}

In Section~\ref{sec:classical} we showed that 
%SBB3: fixed a typo
Crichigno and Prakash's connection between quantum spin chains and characters of the symmetric group can be leveraged in an MPS algorithm for charaters. In ~\cite{crichigno2024quantumspinchainssymmetric} they also prove mappings between spin chains and Kostka numbers and Littlewood-Richardson coefficients. Using similar ideas to characters, we will show that this correspondence can be used to give an MPS algorithm for Kostka numbers.

Kostka numbers appear frequently in combinatorics and representation theory. Let $\lambda \vdash n$. A semistandard Young tableau of shape $\lambda$ is a filling of the boxes in the Young diagram corresponding to $\lambda$ with positive integers between $1$ and $n$ such that the labels are non-decreasing across rows and strictly increasing down columns. The weight of such a tableau is an integer-valued vector of length $n$ such that the $i$th entry is the number of occurrences of $i$. The Kostka number $K_{\lambda, \mu}$ is the number of semistandard Young tableau of shape $\lambda$ and weight $\mu$. Note that $K_{\lambda, \mu}$ is independent of the order of the entries in $\mu$ and thus we assume that $\mu$ is a partition of $n$. A notable special case is $K_{\lambda, (1^n)}$, which is the dimension of the irrep labeled by $\lambda$. More generally, $K_{\lambda,\mu}$ appears in the representation theory of $S_n$ as the multiplicity of the irrep $\lambda$ in the permutation module indexed by $\mu$. In the language of symmetric polynomials, the Kostka numbers encode the change of basis from Schur polynomials to complete symmetric polynomials.

%SBB3: fixe a typo
Crichigno and Prakash use this latter characterization to demonstrate that quantum spin chains encode Kostka numbers as well~\cite{crichigno2024quantumspinchainssymmetric}. We will briefly review their construction here. Again let the initial state $|x_\tau\ra$ correspond to the integer sequences $(n-1,n-2,\ldots,0)$ in the second quantization picture. For commuting variables, the complete symmetric polynomial $h_k$ is given by
\begin{align}
    h_k = \sum_{1 \leq i_1 \cdots \leq i_k} x_{i_1}x_{i_2}\cdots x_{i_k}\ .
\end{align}
For example, in three variables $h_2 = x_1^2 + x_2^2 + x_3^2 + x_1x_2 + x_1x_3 + x_2x_3$. 
%SBB3: fixed a typo
In Crichigno and Prakaksh's construction the commutative variables are replaced with operators $\hat{x}_i := a_{i+1}^\dagger a_i$, where $a$ ($a^\dagger$) is again an annihilation (creation) operator. Then, the non-commutative complete symmetric polynomials take the form
\begin{align}\label{eq:hk}
    \hat{h}_k = \sum_{i_1 > \cdots > i_k} \hat{x}_{i_1}\hat{x}_{i_2}\cdots \hat{x}_{k}\ .
\end{align}
Note that the indices are required to be strictly decreasing, which follows from $\hat{x}_i^2=0$ for any index $i$. Say that $\mu \vdash n$ is a partition of length $m$. We define $\hat{h}_\mu = \prod_{j=1}^{m} \hat{h}_{\mu_j}$ (the operators commute and thus the order of the product does not matter~\cite{crichigno2024quantumspinchainssymmetric}). The authors show that 
\begin{align}
    K_{\lambda,\mu} = \la x_{\lambda+\tau}|\hat{h}_\mu|x_\tau\ra\ .
\end{align}
Note that $\hat{h}_\mu$ moves fermions at most $n$ sites to the right and we can again truncate the chain to $2n$ sites with the initial state $|x_\tau\ra = |1^n0^n\ra$.

Similarly to the MPS algorithm for characters, we will represent the $2n$-qubit state
\begin{align}
    | \psi_\mu\ra = \prod_{i=1}^{m}\hat{h}_{\mu_i}|1^n0^n\ra
\end{align}
by an MPS. We do so via the following lemma.

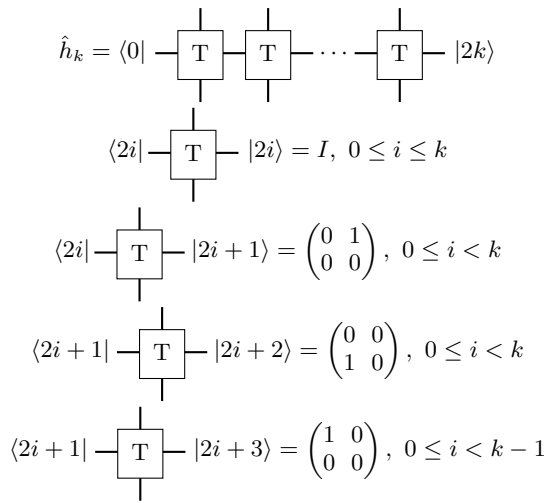
\begin{figure}[h]
    \begin{tikzpicture}
    \node at (-0.7,0.15) {$\hat{h}_k = \langle 0 |$};
    \filldraw[black,thick](0,0.1) -- (0.3,0.1);
    \draw [black] (0.3,-0.25) rectangle (0.9,0.4);
    \node at (0.6,0.1) {T};
    \filldraw[black,thick](0.6,0.4) -- (0.6, 0.7);
    \filldraw[black,thick](0.6, -0.25) -- (0.6, -0.55);
    % next square
    \filldraw[black,thick](0.9,0.1) -- (1.2,0.1);
    \draw [black] (1.2,-0.25) rectangle (1.8,0.4);
    \node at (1.5,0.1) {T};
    \filldraw[black,thick](1.5,0.4) -- (1.5, 0.7);
    \filldraw[black,thick](1.5, -0.25) -- (1.5, -0.55);
    \filldraw[black,thick](1.8,0.1) -- (2.1,0.1);
    \node at (2.4,0.09) {$\cdots$};
    % final square
    \filldraw[black,thick](2.65,0.1) -- (2.95,0.1);
    \draw [black] (2.95,-0.25) rectangle (3.55,0.4);
    \node at (3.25,0.1) {T};
    \filldraw[black,thick](3.25,0.4) -- (3.25, 0.7);
    \filldraw[black,thick](3.25, -0.25) -- (3.25, -0.55);
    \filldraw[black,thick](3.55,0.1) -- (3.85,0.1);
    \node at (4.25,0.12) {$|2k\rangle$};
\end{tikzpicture}
    \begin{tikzpicture}
    \node at (-0.3,0.12) {$\langle 2i |$};
    \filldraw[black,thick](0,0.1) -- (0.3,0.1);
    \draw [black] (0.3,-0.25) rectangle (0.9,0.4);
    \node at (0.6,0.1) {T};
    \filldraw[black,thick](0.6,0.4) -- (0.6, 0.7);
    \filldraw[black,thick](0.6, -0.25) -- (0.6, -0.55);
    \filldraw[black,thick](0.9,0.1) -- (1.2,0.1);
    \node at (2.65,0.12) {$|2i\rangle = I,\  0\leq i \leq k$};
\end{tikzpicture}
    \begin{tikzpicture}
    \node at (-0.3,0.12) {$\langle 2i |$};
    \filldraw[black,thick](0,0.1) -- (0.3,0.1);
    \draw [black] (0.3,-0.25) rectangle (0.9,0.4);
    \node at (0.6,0.1) {T};
    \filldraw[black,thick](0.6,0.4) -- (0.6, 0.7);
    \filldraw[black,thick](0.6, -0.25) -- (0.6, -0.55);
    \filldraw[black,thick](0.9,0.1) -- (1.2,0.1);
    \node at (3.35,0.12) {$|2i+1\rangle = \begin{pmatrix} 0 & 1 \\ 0 & 0 \end{pmatrix},\ 0 \leq i < k$};
\end{tikzpicture}
    \begin{tikzpicture}
    \node at (-0.62,0.12) {$\langle 2i+1 |$};
    \filldraw[black,thick](0,0.1) -- (0.3,0.1);
    \draw [black] (0.3,-0.25) rectangle (0.9,0.4);
    \node at (0.6,0.1) {T};
    \filldraw[black,thick](0.6,0.4) -- (0.6, 0.7);
    \filldraw[black,thick](0.6, -0.25) -- (0.6, -0.55);
    \filldraw[black,thick](0.9,0.1) -- (1.2,0.1);
    \node at (3.35,0.12) {$|2i+2\rangle = \begin{pmatrix} 0 & 0 \\ 1 & 0 \end{pmatrix},\ 0 \leq i < k$};
\end{tikzpicture}
    \begin{tikzpicture}
    \node at (-0.62,0.12) {$\langle 2i+1 |$};
    \filldraw[black,thick](0,0.1) -- (0.3,0.1);
    \draw [black] (0.3,-0.25) rectangle (0.9,0.4);
    \node at (0.6,0.1) {T};
    \filldraw[black,thick](0.6,0.4) -- (0.6, 0.7);
    \filldraw[black,thick](0.6, -0.25) -- (0.6, -0.55);
    \filldraw[black,thick](0.9,0.1) -- (1.2,0.1);
    \node at (3.65,0.12) {$|2i+3\rangle = \begin{pmatrix} 1 & 0 \\ 0 & 0 \end{pmatrix},\ 0 \leq i < k-1$};
\end{tikzpicture}
    \caption{MPO representation of the non-commuting complete symmetric polynomial $\hat{h}_k$. Horizontal lines
    represent virtual indices of dimension $2k+1$. We label basis vectors of $\Cbb^{2k+1}$ by 
     integers $\{0,1,\ldots,2k\}$. The MPO consists of $2n$ copies of $T$. We list all nonzero components of the tensor $T$.}
    \label{fig:mpo_hk}
 \end{figure}

\begin{lemma}
    The non-commuting complete symmetric polynomial $\hat{h}_k$ is an Matrix Product Operator (MPO) with bond dimension $2k+1$.
\end{lemma}
\begin{proof}
    Fig.~\ref{fig:mpo_hk} defines a tensor $T\in\mathbb{C}^2\otimes\mathbb{C}^{2k+1}\otimes\mathbb{C}^{2k+1}\otimes\mathbb{C}^2$, where $\mathbb{C}^2$ represents physical indices and $\mathbb{C}^{2k+1}$ represents virtual indices of the MPO. Taking $2n$ copies of $T$ and contracting the virtual indices gives an operator
    \begin{align}
        \sum_{i_1 > i_2 > \ldots > i_k} \hat{x}_{i_1}\hat{x}_{i_2}\cdots\hat{x}_{i_k}\ .
    \end{align}
    This is exactly $\hat{h}_k$ as defined in Eq.~\ref{eq:hk}.
\end{proof} 

\begin{figure}[h]
  \includegraphics[height=6cm]{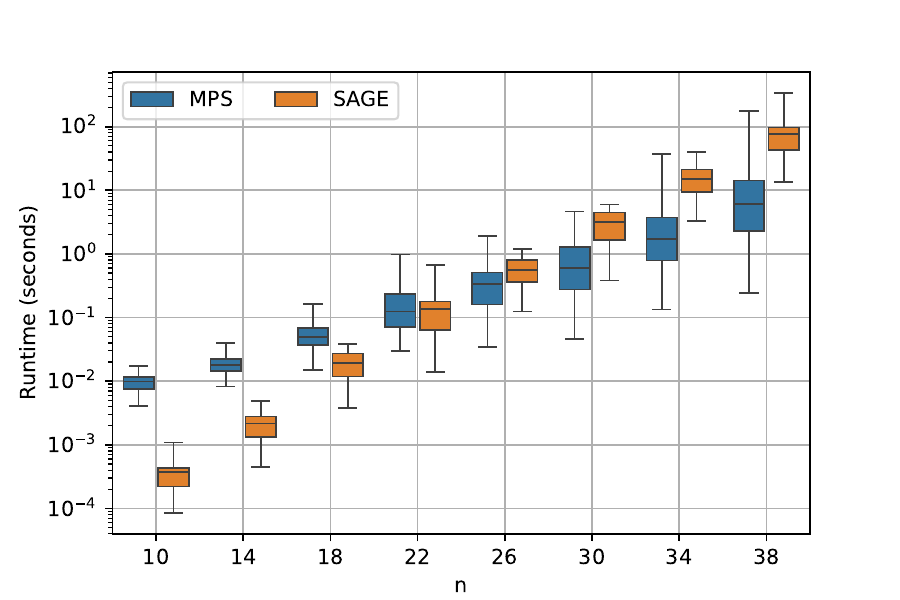}
     \caption{Runtime for computing Kostka numbers $K_{\lambda, \mu}$. Each runtime is, fixing a random $\mu$, for computing the set $\{K_{\lambda, \mu}\}_\lambda$. One hundred random $\mu$'s of length less than $n/3$ were sampled for each $n$. The MPS compression error is set to $\varepsilon = 10^{-12}$ for all the data. Boxes indicate the first and third quartiles and the whiskers denote $1.5$ times the inter-quartile range.}
     \label{fig:kostka_runtime}
\end{figure}

From the bond dimension  of this MPO, we can upper bound the maximum bond dimension of all states during compression as
\begin{align}
    D \leq \prod_{j=1}^{m} (2\mu_j+1)\ .
\end{align}
In running numerics, however, bond dimensions appeared to be significantly lower. 

Denote the compressed MPS state as $| \mathsf{mps}_{\mu} \ra$. Similarly to characters, Kostka numbers take integer values and thus the correct value is recovered as long as 
\begin{align}
    \vert K_{\lambda, \mu} - \la \lambda | \mathsf{mps}_{\mu}\ra | < \frac{1}{2} \ .
\end{align}

We again used $\mathsf{mpnum}$ software for numerical simulations with an error rate $\varepsilon = 10^{-11}$. We compared the MPS algorithm to the $\mathsf{symmetrica}$ library of $\mathsf{SAGE}$, see Fig.~\ref{fig:kostka_runtime}.

For numerical experiments we selected random partitions $\mu$ with a smaller number of parts, fewer than $n/3$. On longer partitions the MPS algorithm appears to require smaller $\varepsilon$ and longer run times. For example, $\mu=(1,\ldots,1)$ proved difficult to compute via the MPS algorithm yet can be quickly computed via the hook length formula. However, in the regime of shorter partitions, our numerics offer evidence that the MPS is competitive and often faster than $\mathsf{Sage}$ for large $n$. It is further worth noting that most partitions are relatively short, the average being $\Theta(\sqrt{n}\log(n))$~\cite{erdos1941distribution}. Our implementation, as well as one for skew Kostka numbers (a generalization), can be found at~\cite{character_builder}. 

An interesting question, which we leave for future work, is if Littlewood-Richardson coefficients can be recovered by a non-trivial MPS algorithm. Doing so via a mapping to spin chains would require MPO constructions for non-commutative Schur polynomials~\cite{crichigno2024quantumspinchainssymmetric}. Schur polynomials do not admit an obvious modular construction unlike $\prod_{j=1}^{c(g)}J_{\eta_j}$ and $\prod_{j=1}^m \hat{h}_{\mu_j}$ (from the commuting products). This would seem to complicate MPO constructions.

% Let $\nu \vdash m$ and $\lambda \vdash m+n$. Further assume that $\lambda_i - \mu_i \geq 0$ for all $1 \le i \le \ell(\nu)$ (graphically this corresponds to all of the boxes in $\eta$ being contained in $\Lambda $ as well). Then, $\lambda \backslash\nu = (\lambda_1 - \nu_1, \lambda_2 - \nu_2, \cdots)$ is the skew diagram formed by removing the boxes of $\nu$ from $\lambda$. Similarly to normal tableaux, we say that a filling of the boxes with integers between $1$ and $n$ is semistandard if the entries are non-decreasing across rows and strictly increasing down columns. Let $\mu \vdash n$. Then, the skew Kostka number $K_{\lambda \backslash \nu, \mu}$ is the number of semistandard skew tableau of shape $\lambda \backslash \nu$ and weight $\mu$.

% Crinchigno and Prakaksh show that 
% \begin{align}
%     K_{\lambda \backslash \nu, \mu} = \la\lambda| \hat{h}_{\mu} | \nu \ra\ .
% \end{align}
% Thus, our MPO implementation of $\hat{h}_k$ can be used to compute skew Kostka numbers as well. Note that the number of sites needs to be increased from $2n$ to $2(m+n)$ to account for the fact that $\lambda \vdash m+n$.

%SBB: new section
%DG: added acknowledgment
\begin{acknowledgments}       
        This research was supported in part by Perimeter Institute for Theoretical Physics. Research at Perimeter Institute
is supported by the Government of Canada through the Department of Innovation, Science, and Economic
Development, and by the Province of Ontario through the Ministry of Colleges and Universities. DG is a fellow of the Canadian Institute for Advanced Research, in the quantum information science program. LS was supported by IBM through the Illinois-IBM Discovery Accelerator Institute.
\end{acknowledgments}

\bibliographystyle{unsrt}
\bibliography{mybib}
\end{document}